\documentclass[journal]{IEEEtran}
\IEEEoverridecommandlockouts
\usepackage[hidelinks,colorlinks,allcolors=blue]{hyperref}
\usepackage{cite}
\usepackage{amsmath,amssymb,amsfonts}
\usepackage{algorithmic}
\usepackage{graphicx}
\usepackage{textcomp}
\usepackage{xcolor}
\usepackage{amsfonts}
\usepackage{booktabs}
\usepackage{siunitx}
\usepackage{bbm}
\usepackage{multirow}
\usepackage{algorithm}
\usepackage{comment}
\usepackage{balance}
\usepackage[caption=false,font=footnotesize]{subfig}
\usepackage{amsthm}

\newtheorem{proposition}{Proposition}   
\newtheorem{remark}{Remark}

\usepackage{fancyhdr}

\fancypagestyle{firstpagecopyright}{%
  \fancyhf{} 
  \fancyfoot[C]{\parbox{\textwidth}{\footnotesize © 2026 IEEE. Personal use of this material is permitted. Permission from IEEE must be obtained for all other uses, in any current or future media, including reprinting/republishing this material for advertising or promotional purposes, creating new collective works, for resale or redistribution to servers or lists, or reuse of any copyrighted component of this work in other works.}}
}

\begin{document}
\markboth{Accepted to IEEE Signal Processing Letters}%
{Faghih Niresi \MakeLowercase{\textit{et al.}}}

\title{Graph Signal Separation with Learnable Spectral Filters\\

}
\author{Keivan Faghih Niresi, Dorina Thanou, and Olga Fink

\thanks{Keivan Faghih Niresi and Olga Fink are with Intelligent Maintenance and Operations Systems (IMOS) Laboratory, EPFL, Lausanne, Switzerland
(e-mail: keivan.faghihniresi@epfl.ch, olga.fink@epfl.ch).}
\thanks{Dorina Thanou is with Signal Processing Laboratory 4 (LTS4), EPFL, Lausanne, Switzerland
(e-mail: dorina.thanou@epfl.ch).}

}

\maketitle
\thispagestyle{firstpagecopyright}

\begin{abstract}
Separating multiple graph signals from a single observed mixture is an inherently ill-posed problem that traditionally relies on restrictive and handcrafted priors. This letter addresses this challenge by proposing an unsupervised learnable spectral filtering framework.  Our approach reconstructs latent components by passing a fixed random input through learnable spectral filters, operating within the low-frequency eigenspace of each source-specific graph Laplacian. The architecture implicitly biases the recovered signals toward smooth patterns by confining reconstruction to these low-frequency subspaces. This acts as a structural prior, establishing a principled bridge between classical graph spectral analysis and modern neural decomposition. Numerical experiments confirm that this framework successfully isolates individual sources using solely the observed mixture and the underlying graph topology.
\end{abstract}


\section{Introduction}

Blind source separation (BSS) is a fundamental problem in classical signal processing, aiming to recover latent and statistically independent source signals from their observed mixtures \cite{cardoso1998blind}. While classical BSS has found extensive applications across various domains \cite{musluoglu2025distributed, guo2024single}, it traditionally relies on a critical assumption: the availability of multiple mixtures. Typically, this requires the number of observations to be at least equal to the number of underlying sources \cite{comon2010handbook}. Furthermore, classical BSS is largely designed for signals sampled on regular domains. However, data is increasingly characterized by measurements that are inherently interconnected through irregular and relational structures, naturally represented by graphs. Graph signal processing (GSP) \cite{shuman2013emerging, ortega2018graph, ortega2022introduction, leus2023graph} accommodates such irregular domains by leveraging the graph Laplacian or related shift operators to define notions like smoothness, band-limitation, and filtering in the spectral domain. The success of GSP in numerous applications such as graph signal recovery \cite{torkamani2020statistical, chen2016signal}, graph signal sampling \cite{tanaka2020sampling}, and graph learning \cite{dong2019learning, javaheri2025time, niresi2024informed} has recently motivated the extension of BSS methodologies to graph-structured data. These emerging graph signal separation approaches can be broadly categorized according to the number of available observations.

The first category assumes access to multiple mixtures and extends classical BSS principles to graph signals. Representative examples include methods that jointly diagonalize graph auto‑covariance and higher‑order cumulant matrices under the assumption of known underlying graphs \cite{blochl2010second,miettinen2021graph,miettinen2020blind}, and a recent dynamic‑weight variant that adaptively balances graph structure and non‑Gaussianity \cite{zhang2025blind}. Several extensions further relax the assumption of known graph topologies by jointly estimating the graph structures and source signals from multiple observations \cite{einizade2021simultaneous, einizade2022unified}. Despite these advances, all such methods fundamentally require the availability of multiple mixed observations. The second category considers the more challenging single-mixture setting, where only one observation is available. Since classical statistical separation principles become insufficient in this regime, existing methods typically incorporate additional structural assumptions, graph priors, or generative models to regularize the recovery process \cite{iglesias2018demixing, garcia2020blind}. Nevertheless, recovering multiple graph signals from a single observation is a highly ill-posed problem, and without known graphs, jointly inferring both the topologies and the source signals from a single mixture is generally impossible without further strong assumptions.

A natural and prominent class of approaches for resolving this single-mixture separation relies on (possibly convex) regularization \cite{yarandi2023closed, mohammadi2023graph}. In these methods, each source is recovered by minimizing a data-fidelity term coupled with graph-smoothness penalties defined with respect to the corresponding graph Laplacians \cite{yarandi2024new, mohammadi2023graph}. Such formulations are advantageous due to their mathematical interpretability. However, their performance remains bound to the specific penalty form and careful weight tuning. By prescribing spectral characteristics a priori, these methods lack the adaptability required when latent components exhibit heterogeneous frequency content, or when the optimal spectral allocation cannot be adequately captured by predefined convex penalties.

This letter revisits graph signal separation from a spectral learning perspective and proposes an unsupervised spectral filtering framework to decompose a single noisy observation into multiple graph-specific components. Each source is modeled by a learnable spectral response acting on the low-frequency eigenspace of its associated graph Laplacian. Instead of directly optimizing source signals, the method learns graph-specific spectral coefficients modulating a fixed random latent input. This formulation offers two advantages. First, it requires no mixture-source training pairs, as all parameters are inferred directly from a single noisy observation. Second, it avoids black-box modeling by defining an explicit graph-spectral operator clearly interpreted via Laplacian eigenmodes. This preserves classical spectral graph processing interpretability while benefiting from the flexibility of learned spectral responses. Importantly, this formulation differs from prior graph-BSS approaches that rely on statistical independence, mutual information, or explicit models, and from graph separation methods that impose fixed spectral priors; instead, it learns spectral biases directly from the data.

The main contributions of this letter are threefold. First, we formulate single-observation graph signal separation within a spectral learning framework. Second, we propose a learnable spectral filter in which each component is represented by a low-dimensional learnable filter acting on the eigenbasis of its associated graph Laplacian and driven by a fixed random latent input. Third, we derive an unsupervised reconstruction objective that enables the direct estimation of the spectral coefficients from the observed mixture, enabling source decomposition without the need for paired supervision.

\section{Related Background}

\subsection{Graph Signal Separation}
Graph signal separation considers the recovery of multiple latent graph signals from an observed mixture when the constituent components exhibit distinct structural priors on different graphs. Let $\mathcal{G}_p = (\mathcal{V},\mathcal{E}_p)$ for $p=1,\dots,P$ denote a collection of graphs sharing a common vertex set $\mathcal{V}$ with $|\mathcal{V}|=N$ and having edge sets $\mathcal{E}_p$, and let $\mathbf{L}_p \in \mathbb{R}^{N \times N}$ be the Laplacian matrix associated with $\mathcal{G}_p$. The observed signal $\mathbf{m} \in \mathbb{R}^N$ is formed by an additive superposition of $P$ latent components:
\begin{equation} \label{eq:mixture}
\mathbf{m} = \sum_{p=1}^{P} \mathbf{x}_p^\star + \boldsymbol{\epsilon},
\end{equation}
where $\mathbf{x}_p^\star \in \mathbb{R}^N$ denotes the $p$-th latent graph signal and $\boldsymbol{\epsilon}$ is additive noise. A key assumption in graph signal separation is that each latent component is smooth with respect to its corresponding graph, meaning its energy is concentrated in the low-frequency spectral subspace of the corresponding Laplacian. To ensure identifiability, each component is constrained to have zero DC value, i.e., $\mathbf{1}^\top \mathbf{x}_p = 0$ for all $p$. A classical formulation that promotes this smoothness is given by \cite{mohammadi2023graph}:
\begin{equation} \label{eq:convex_reg}
\begin{aligned}
\min_{\{\mathbf{x}_p\}_{p=1}^{P}} \quad &
\frac{1}{2}\left\|\mathbf{m} - \sum_{p=1}^{P}\mathbf{x}_p\right\|_2^2
+
\sum_{p=1}^{P}\gamma_p\, \mathbf{x}_p^\top \mathbf{L}_p \mathbf{x}_p \\
\text{s.t.} \quad & \mathbf{1}^\top \mathbf{x}_p = 0,\ \text{for} \; p= 1,\dots,P
\end{aligned}
\end{equation}
where $\gamma_p > 0$ controls the strength of the smoothness prior for each component. While effective, this formulation imposes spectral responses that are fixed \textit{a priori} through the chosen regularization, thereby limiting expressivity when the optimal frequency allocation is unknown or varies across components.

\subsection{Spectral Graph Filtering}
Spectral methods transform graph signals by projecting them onto the eigenbasis of the graph Laplacian. For each graph $\mathcal{G}_p$, consider the eigendecomposition $\mathbf{L}_p = \mathbf{U}_p \boldsymbol{\Lambda}_p \mathbf{U}_p^\top$, where $\mathbf{U}_p$ is an orthonormal matrix of eigenvectors and $\boldsymbol{\Lambda}_p$ is a diagonal matrix containing the eigenvalues in increasing order. A spectral graph filter is defined by a frequency response function $h_p(\cdot)$ and acts on a signal $\mathbf{x}$ as \cite{isufi2024graph}:
\begin{equation}
\mathbf{y} = \mathbf{U}_p\, h_p(\boldsymbol{\Lambda}_p)\, \mathbf{U}_p^\top \mathbf{x},
\end{equation}
where $h_p(\boldsymbol{\Lambda}_p)$ is a diagonal matrix. This operator modulates graph frequencies individually. In signal separation, spectral filtering can be highly effective because it enables the isolation of sources based on their distinct low-frequency support across different graphs. To overcome the limitations of fixed transfer functions, learnable spectral filtering operators can adapt the frequency response directly from observed data.

\section{Proposed Method}

\subsection{Problem Formulation}
The proposed method addresses the recovery of the $P$ latent graph signals $\{\mathbf{x}_p^\star\}_{p=1}^{P}$ from the single noisy mixture $\mathbf{m}$ defined in Eq.~\eqref{eq:mixture}. The goal is to estimate a set of components $\{\hat{\mathbf{x}}_p\}_{p=1}^{P}$ such that $\mathbf{m} \approx \sum_{p=1}^{P}\hat{\mathbf{x}}_p$, constrained by the spectral structure of their corresponding graph Laplacians, which are known a priori. In contrast to supervised approaches, the method operates without paired training data. Instead, it directly learns graph spectral filters from a single observation, leveraging the underlying graph structure as an inductive bias.

\subsection{Learnable Spectral Filter and Optimization}
To enforce a low-dimensional subspace, we truncate the spectral representation. Let $\mathbf{U}_p^{(k_p)} \in \mathbb{R}^{N \times k_p}$ denote the matrix formed by the first \textit{non-constant} $k_p$ eigenvectors of $\mathbf{L}_p$ corresponding to the smallest non-zero eigenvalues. This truncation retains only the low‑frequency basis vectors where smooth graph signals concentrate their energy~\cite{shuman2013emerging}, acting as an implicit band‑limiting prior. A key component of the proposed architecture is a fixed and randomly initialized latent input vector $\mathbf{z} \in \mathbb{R}^N$. This vector remains unchanged during optimization and serves as a shared excitation signal across all components. The $p$-th learnable spectral branch maps this latent input into the corresponding estimated source component:
\begin{equation} \label{eq:branch_output}
\hat{\mathbf{x}}_p(\boldsymbol{\alpha}_p)
=
\mathbf{U}_p^{(k_p)}
\operatorname{diag}(\boldsymbol{\alpha}_p)
\left(\mathbf{U}_p^{(k_p)}\right)^\top \mathbf{z},
\quad p=1,\dots,P,
\end{equation}
where $\boldsymbol{\alpha}_p \in \mathbb{R}^{k_p}$ is a trainable vector of spectral coefficients. The projection $\left(\mathbf{U}_p^{(k_p)}\right)^\top \mathbf{z}$ extracts the low-frequency components of the latent input, which are subsequently reweighted by $\operatorname{diag}(\boldsymbol{\alpha}_p)$, and $\mathbf{U}_p^{(k_p)}$ maps the filtered signal back to the vertex domain. Since both the latent input $\mathbf{z}$ and the Laplacian eigenvectors are fixed, the only trainable parameters are the spectral coefficient vectors $\{\boldsymbol{\alpha}_p\}_{p=1}^{P}$. These parameters are learned by minimizing the reconstruction error of the mixture:
\begin{equation}
\{\hat{\boldsymbol{\alpha}}_p\}_{p=1}^{P}
=
\operatorname*{arg\,min}_{\{\boldsymbol{\alpha}_p\}_{p=1}^{P}}
\left\|
\mathbf{m}
-
\sum_{p=1}^{P}
\hat{\mathbf{x}}_p(\boldsymbol{\alpha}_p)
\right\|_2^2.
\end{equation}
Through this optimization, the mixture energy is distributed across the $P$ branches. As illustrated in Fig.~\ref{fig:snf}, each branch processes a shared latent input $\mathbf{z}$ by projecting it onto the low-frequency eigenspace of its respective graph. It then applies learnable spectral gains $\boldsymbol{\alpha}_p$ to reconstruct the corresponding source. All parameters can be updated using standard gradient-based methods.

\begin{figure} \centering \includegraphics[width=\linewidth]{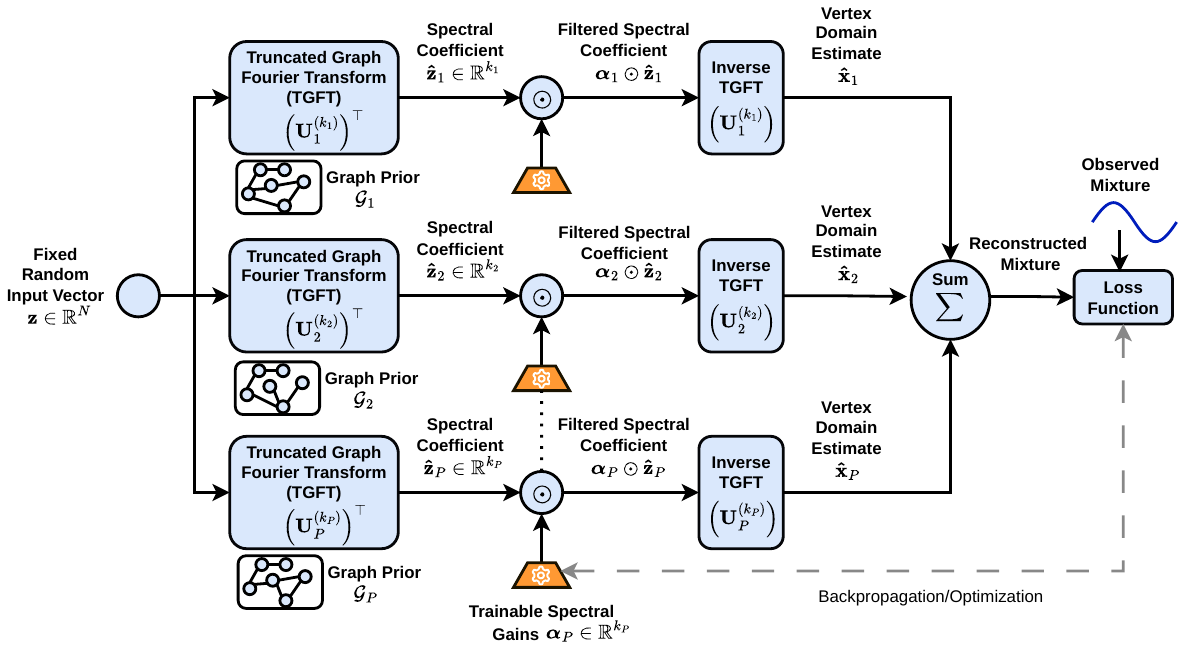} 
\caption{Illustration of the proposed learnable spectral filtering method for graph signal separation. The operator $\odot$ denotes the Hadamard product.} \label{fig:snf} \end{figure}

\subsection{Selection of Spectral Cutoff}

In the proposed framework, each graph signal is projected onto the first $k_p$ \textit{non-constant} Laplacian eigenvectors to capture its smooth components. The value of $k_p$ is determined adaptively for each graph $\mathcal{G}_p$: we select all eigenvectors with eigenvalues below a fraction $\lambda_{\text{ratio}}$ of the largest non-zero eigenvalue. This criterion ensures that the learned spectral filters focus on the low-frequency content of the signal, thereby reducing computational complexity and mitigating noise, under the assumption that the underlying graph signals are smooth and dominated by low-frequency components.\footnote{Signals with significant high-frequency content may require a larger $\lambda_{\text{ratio}}$.}

\subsection{Computational Complexity}
We analyze the computational complexity of the proposed method. The eigendecomposition of the Laplacians is performed once offline. Using efficient iterative methods such as the Lanczos algorithm \cite{lanczos1950iteration}, computing the first $k_p$ eigenvectors requires approximately $\mathcal{O}(k_p |\mathcal{E}_p|)$ operations, where $|\mathcal{E}_p|$ is the number of edges in $\mathcal{G}_p$. For simplicity, we assume $k_1 = k_2 = \dots = k_P = k$. During  optimization, each forward pass involves projecting $\mathbf{z}$ onto the truncated basis ($\mathcal{O}(kN)$), scaling by $\boldsymbol{\alpha}_p$ ($\mathcal{O}(k)$), and mapping back to the vertex domain ($\mathcal{O}(kN)$). Across $P$ branches, the per-iteration complexity therefore scales as $\mathcal{O}(PkN)$. Since $k \ll N$, the parameter space is limited to $k$ per branch, resulting in overall memory and time complexity that scale linearly with the number of nodes $N$. By contrast, Eq.~\eqref{eq:convex_reg} is a convex quadratic program over all \(P N\) nodal variables, leading to a much higher-dimensional optimization problem (its per‑iteration cost depends on the chosen solver and can be comparable to the proposed method for first‑order algorithms).

\subsection{Linear Representation and Spectral Interpretation}
The proposed learnable spectral filtering model admits an equivalent linear representation with respect to the trainable spectral coefficients for fixed graph bases, providing insight into the underlying estimation problem. For simplicity, we assume a common truncation level across all graphs, i.e., $k_1 = k_2 = \dots = k_P = k$. For the $p$-th branch, let the fixed latent input be $\mathbf{z} \in \mathbb{R}^N$, and denote its projection onto the truncated graph Fourier basis by $\mathbf{\hat{z}}_p = (\mathbf{U}_p^{(k)})^\top \mathbf{z} \in \mathbb{R}^k$. The component estimate can then be written as $\hat{\mathbf{x}}_p = \mathbf{U}_p^{(k)} \, \mathrm{diag}(\boldsymbol{\alpha}_p)\, \mathbf{\hat{z}}_p$. Using the identity $\mathrm{diag}(\boldsymbol{\alpha}_p)\mathbf{\hat{z}}_p = \mathrm{diag}(\mathbf{\hat{z}}_p)\boldsymbol{\alpha}_p$, this becomes $\hat{\mathbf{x}}_p = \mathbf{U}_p^{(k)} \, \mathrm{diag}(\mathbf{\hat{z}}_p)\boldsymbol{\alpha}_p$. Defining $\mathbf{B}_p = \mathbf{U}_p^{(k)} \, \mathrm{diag}(\mathbf{\hat{z}}_p)$, the estimate simplifies to the linear form $\hat{\mathbf{x}}_p = \mathbf{B}_p \boldsymbol{\alpha}_p$. By stacking all coefficient vectors into $\boldsymbol{\alpha}$ and defining $\mathbf{A} = [\mathbf{B}_1, \dots, \mathbf{B}_P]$, the reconstructed mixture $\hat{\mathbf{m}} = \sum_{p=1}^{P} \mathbf{B}_p \boldsymbol{\alpha}_p$ reduces to the linear system $
\hat{\mathbf{m}} = \mathbf{A} \boldsymbol{\alpha},
$
and the estimation problem becomes:
\begin{equation}
\min_{\boldsymbol{\alpha}} \| \mathbf{m} - \mathbf{A}\boldsymbol{\alpha} \|_2^2.
\end{equation}

This reformulation shows that, for fixed graph bases and input, the proposed method reduces to a linear inverse problem \cite{niresi2025rins} that is convex in the spectral coefficients. Since the matrix $\mathbf{A}$ is constructed from graph-dependent spectral bases and a fixed random excitation, the recovered signals are constrained to lie in a low-dimensional subspace spanned by scaled eigenvectors. From this perspective, the method can be interpreted as a spectral parameterization of the reconstruction problem, rather than a direct optimization over the source signals $\mathbf{x}_p$. In contrast to traditional regularizers such as the graph Sobolev norm \cite{giraldo2022reconstruction} or Dirichlet energy, which impose predefined smoothness priors, the proposed formulation learns the spectral gains $\boldsymbol{\alpha}_p$ directly from the data, enabling adaptive rather than predefined spectral responses for each component. Under a white Gaussian noise model and the assumption that the true sources lie within the truncated spectral subspaces, the least-squares estimate \(\hat{\boldsymbol{\alpha}}\) attains the Cram\'er--Rao lower bound, making it an efficient estimator; a detailed derivation is provided in the supplementary material.

\section{Numerical Results}

We evaluate the proposed learnable spectral filter (GSS‑LSF) against the convex smoothness baseline (GSS‑Smooth \cite{mohammadi2023graph}) and the ratio‑based method (GSS‑Ratio \cite{yarandi2024new}), which promotes smoothness on each source’s own graph while penalizing smoothness on the other graphs. Performance is measured using per-source SNR and average SNR. Hyperparameters are selected via grid search in the first trial and then fixed for all subsequent runs for both methods.\footnote{Cases 1 \& 2 use $\lambda_{\text{ratio}} = 0.1$ to enforce stronger smoothness, whereas Case 3 uses $\lambda_{\text{ratio}} = 0.5$ to allow higher-frequency components; sensitivity analysis is provided in the supplementary material.} All reported results are averaged over three trials.

\subsection{Experiment 1: Two-Source Mixture}

In this experiment, we separate two graph signals from a noisy
mixture using two graph sizes: $N = 250$ and $N = 350$. Two independent graphs are generated: a connected random geometric graph and a random 4-regular graph. The first source signal (S1) is generated as a random linear combination of the first two nontrivial eigenvectors of the Laplacian of the first graph, while the second source (S2) is obtained from the first five nontrivial eigenvectors of the second graph Laplacian. Both signals are normalized to have zero mean and unit variance.  Fig. \ref{fig:case1} summarizes the results across input SNR levels of \{0, 5, 10, 15, 20\} dB. The proposed GSS-LSF method consistently achieves higher output SNR than both GSS-Smooth and GSS-Ratio baselines for both graph sizes across all noise levels. Moreover, as the input SNR increases, GSS-LSF exhibits an approximately linear improvement in performance, effectively recovering both source signals from the noisy mixture. In particular, S1 is generally reconstructed more accurately than S2, due to its lower spectral complexity.

\begin{figure}[h]
    \centering
    \includegraphics[width=\linewidth]{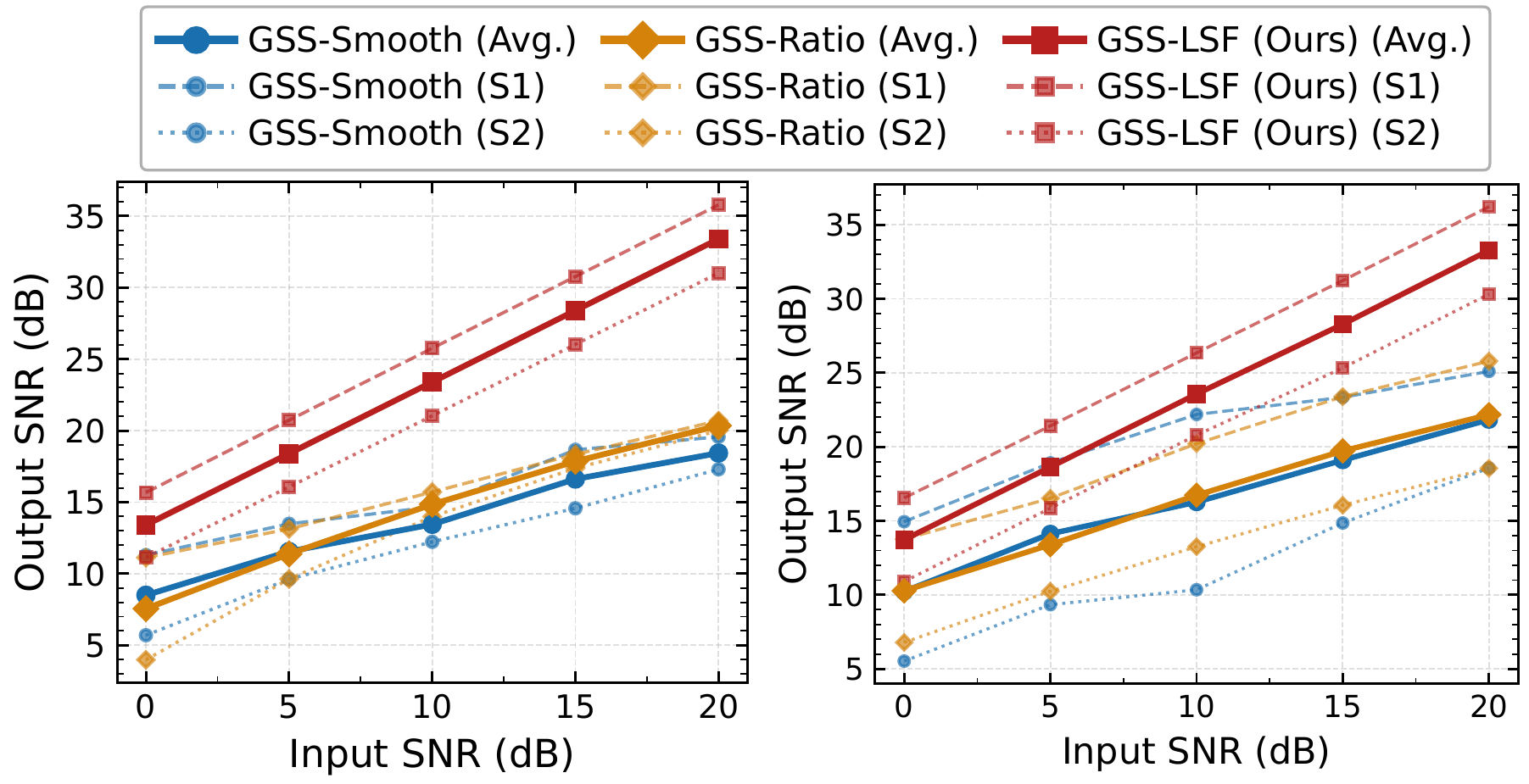}
    \caption{Comparison of output SNR versus input SNR for $N = 250$ (left) and $N = 350$ (right). The plot shows the performance of the proposed  GSS-LSF and other methods. For each method, results are presented as the average (Avg.), as well as individual sources S1 and S2.}

    \label{fig:case1}
\end{figure}

\subsection{Experiment 2: Four-Source Mixture}

In this experiment, we evaluate scalability in a more challenging setting with $P = 4$ mixed graph signals and graph sizes $N = 250$ and $N = 350$. We generate four independent graphs: one connected random geometric graph and three random 4-regular graphs. Each source signal is constructed as a random linear combination of low-frequency Laplacian eigenvectors, using the first $2$, $4$, $6$, and $8$ nontrivial eigenvectors, respectively, resulting in progressively less smooth signals. All signals are normalized, and the observed mixture is obtained by summing the sources and adding Gaussian noise with $\sigma = 0.2$. The results in Table~\ref{tab:case2_4mix} demonstrate that as mixture complexity increases, the performance gap widens. For $N = 250$, the proposed GSS-LSF achieves an average SNR of $25.89$ dB, substantially outperforming GSS-Smooth ($12.47$ dB) and GSS-Ratio ($14.74$ dB). For $N = 350$, the proposed method reaches $25.29$ dB, while GSS-Smooth and GSS-Ratio achieve $14.77$ dB and $15.12$ dB, respectively. While both baselines struggle to separate multiple sources, GSS-LSF maintains high recovery accuracy across all components.

\begin{table}[ht]
\centering
\caption{Performance for four-source mixtures}
\label{tab:case2_4mix}
\resizebox{0.9\linewidth}{!}{%
\begin{tabular}{c c c c c c c}
\toprule
\textbf{Nodes} & \textbf{Method} & \textbf{SNR$_1$} & \textbf{SNR$_2$} & \textbf{SNR$_3$} & \textbf{SNR$_4$} & \textbf{Avg. SNR} \\
\midrule
250 & GSS-Smooth & 16.72 & 11.69 & 9.82  & 11.65 & 12.47 \\
250 & GSS-Ratio  & 16.25 & 14.16 & 13.70 & 14.85 & 14.74 \\
250 & GSS-LSF (Ours)   & \textbf{31.06} & \textbf{22.86} & \textbf{25.68} & \textbf{23.96} & \textbf{25.89} \\
\midrule
350 & GSS-Smooth & 21.47 & 13.09 & 12.47 & 12.04 & 14.77 \\
350 & GSS-Ratio  & 17.07 & 15.68 & 14.76 & 12.96 & 15.12 \\
350 & GSS-LSF (Ours)    & \textbf{29.42} & \textbf{20.32} & \textbf{24.18} & \textbf{27.22} & \textbf{25.29} \\
\bottomrule
\end{tabular}%
}
\end{table}

\subsection{Experiment 3: Smooth (non-Bandlimited) Source Mixture}

We evaluate a two‑source mixture with $N = 250$ and $N = 350$ using two random geometric graphs. Each source is generated as $\mathbf{x}_p^\star = \mathbf{U}_p e^{-\tau \boldsymbol{\Lambda}_p} \mathbf{c}_p$ with $\tau=10$, yielding smooth but non‑bandlimited signals. The mixture is corrupted by Gaussian noise ($\sigma = 0.2$). As shown in Table~\ref{tab:case2_2mix}, the proposed GSS‑LSF consistently outperforms both baselines. For $N=250$, GSS‑LSF attains $26.47$ dB average SNR, compared to $19.29$ dB (GSS‑Smooth) and $19.46$ dB (GSS‑Ratio). For $N=350$, GSS‑LSF reaches $27.60$ dB, while GSS‑Smooth and GSS‑Ratio achieve $20.08$ dB and $20.13$ dB, respectively. The gain exceeds $7$ dB in both cases, confirming the advantage of the learnable spectral filter.

\begin{table}[ht]
\centering
\caption{Performance for two-source (non-Bandlimited) mixtures}
\label{tab:case2_2mix}
\resizebox{0.8\linewidth}{!}{%
\begin{tabular}{c c c c c}
\toprule
\textbf{Nodes} & \textbf{Method} & \textbf{SNR$_1$} & \textbf{SNR$_2$} & \textbf{Avg. SNR} \\
\midrule
250 & GSS-Smooth          & 18.97 & 19.60 & 19.29 \\
250 & GSS-Ratio           & 19.31 & 19.61 & 19.46 \\
250 & GSS-LSF (Ours)    & \textbf{25.40} & \textbf{27.53} & \textbf{26.47} \\
\midrule
350 & GSS-Smooth          & 20.39 & 19.79 & 20.08 \\
350 & GSS-Ratio           & 20.42 & 19.84 & 20.13 \\
350 &  GSS-LSF (Ours)    & \textbf{27.57} & \textbf{27.63} & \textbf{27.60} \\
\bottomrule
\end{tabular}%
}
\end{table}

\subsection{Experiment 4: Real-World Sensor Data Mixture}

The proposed method is evaluated on a real‑world district heating dataset containing 14 temperature and 14 pressure sensors installed on pipes, consumers, and pumps~\cite{niresi2026virtual}. The signals are normalized, mixed, and corrupted by Gaussian noise ($\sigma = 0.2$). Graphs for each modality are constructed by connecting each sensor to its five nearest neighbors based on the Euclidean distances between sensor pairs. Table~\ref{tab:real_data} reports the reconstruction SNRs for all methods. The proposed GSS‑LSF achieves an average SNR of $19.38$~dB, outperforming both GSS‑Smooth ($17.73$~dB) and GSS‑Ratio ($17.49$~dB). Notably, GSS‑Ratio improves the temperature recovery ($18.76$~dB) over GSS‑Smooth, but its pressure separation is less accurate, leading to a lower average.

\begin{table}[h]
\centering
\caption{Performance for real‑world sensor data}
\label{tab:real_data}
\resizebox{0.9\linewidth}{!}{%
\begin{tabular}{l c c c}
\toprule
\textbf{Method} & \textbf{Temperature SNR} & \textbf{Pressure SNR} & \textbf{Avg. SNR} \\
\midrule
GSS-Smooth         & 17.03 & 18.43 & 17.73 \\
GSS-Ratio          & 18.76 & 16.23 & 17.49 \\
GSS‑LSF (Ours)   & \textbf{19.06} & \textbf{19.70} & \textbf{19.38} \\
\bottomrule
\end{tabular}
}
\end{table}

\section{Conclusion}

This letter proposes an unsupervised learnable spectral filtering framework for single-mixture graph signal separation. By utilizing a fixed random input, each source is reconstructed through learnable spectral coefficients applied to a truncated graph Fourier basis, allowing for optimization without paired training data. Our approach effectively learns task-adaptive spectral responses while maintaining graph smoothness, with experimental results confirming its capability for effective source recovery.

\bibliographystyle{IEEEtran}
\bibliography{IEEEabrv,References}

\clearpage

\markboth{Supplementary Materials for Graph Signal Separation with Learnable Spectral Filters}%
{Faghih Niresi \MakeLowercase{\textit{et al.}}: Supplementary Material}

\section*{Identifiability Condition}

\begin{proposition}
\label{prop:ident}
Let $\mathbf{U}_p^{(k_p)}\in \mathbb{R}^{N\times k_p}$ contain the first $k_p$ non-constant eigenvectors of the graph Laplacian $\mathbf{L}_p$.
For a fixed random latent input $\mathbf{z}\in\mathbb{R}^N$, define
\[
\mathbf{B}_p = \mathbf{U}_p^{(k_p)} \operatorname{diag}\!\big((\mathbf{U}_p^{(k_p)})^\top \mathbf{z}\big),\qquad p=1,\dots,P,
\]
and assemble the global dictionary $\mathbf{A}=[\mathbf{B}_1\ \cdots\ \mathbf{B}_P]\in\mathbb{R}^{N\times \sum_p k_p}$.
The spectral coefficient vector $\boldsymbol{\alpha}=[\boldsymbol{\alpha}_1^\top,\dots,\boldsymbol{\alpha}_P^\top]^\top$ is uniquely recoverable from a noiseless mixture $\mathbf{m}=\mathbf{A}\boldsymbol{\alpha}$ if and only if $\operatorname{rank}(\mathbf{A}) = \sum_{p=1}^{P} k_p$.

\end{proposition}

\begin{proof}
If two coefficient vectors $\boldsymbol{\alpha}$ and $\boldsymbol{\alpha}'$ satisfy $\mathbf{A}\boldsymbol{\alpha}=\mathbf{A}\boldsymbol{\alpha}'$, then $\mathbf{A}(\boldsymbol{\alpha}-\boldsymbol{\alpha}')=\mathbf{0}$.
When $\mathbf{A}$ has full column rank, its nullspace is trivial, forcing $\boldsymbol{\alpha}=\boldsymbol{\alpha}'$.
Conversely, if $\operatorname{rank}(\mathbf{A}) < \sum_p k_p$, there exists a non-zero $\Delta\boldsymbol{\alpha}\in\ker(\mathbf{A})$ (where $\ker(\mathbf{A})$ denotes the nullspace of $\mathbf{A}$), yielding a distinct coefficient vector that produces the same mixture, so identifiability is lost.
Because $\hat{\mathbf{x}}_p = \mathbf{B}_p\boldsymbol{\alpha}_p$, uniqueness of $\boldsymbol{\alpha}$ directly implies uniqueness of the estimated source components.
\end{proof}

\begin{remark}
When the latent input $\mathbf{z}$ is drawn from a continuous distribution (e.g., i.i.d.\ Gaussian) and the truncated eigenspaces $\operatorname{span}(\mathbf{U}_p^{(k_p)})$ are sufficiently distinct (e.g., no shared low‑frequency eigenvectors), the matrix $\mathbf{A}$ attains full column rank.
\end{remark}
\vspace{-1em}
\begin{remark}
Excluding the constant eigenvector (associated with zero eigenvalue) from every $\mathbf{U}_p^{(k_p)}$ is essential since the all-ones vector would belong to all spatial bases, immediately creating linear dependence and reducing the rank of $\mathbf{A}$.
\end{remark}
\vspace{-1em}
\section*{Cram\'er--Rao Bound}

For simplicity, we assume that all sources share the same truncation order, i.e., $k_p = k$ for $p = 1,\dots,P$, so that each $\mathbf{B}_p\in\mathbb{R}^{N\times k}$ and $\mathbf{A}\in\mathbb{R}^{N\times Pk}$. Assume that $\mathbf A$ has full column rank, so that $Pk \le N$ and $\mathbf A^\top\mathbf A$ is invertible. The additive noise is white Gaussian, and the log-likelihood function (up to additive constants, denoted by $\mathrm{const}$) is
\begin{equation}
\ell(\boldsymbol{\alpha}) = -\frac{1}{2\sigma^2}\|\mathbf m - \mathbf A\boldsymbol{\alpha}\|_2^2 + \mathrm{const}.
\end{equation}
The Fisher information matrix is
\begin{equation}
\mathbf{F} = \frac{1}{\sigma^2}\mathbf A^\top\mathbf A.
\end{equation}
Hence, the Cram\'er--Rao bound for any unbiased estimator of $\boldsymbol{\alpha}$ is
\begin{equation}
\operatorname{Cov}(\hat{\boldsymbol{\alpha}}) \succeq \mathbf{F}^{-1} = \sigma^2(\mathbf A^\top\mathbf A)^{-1},
\end{equation}
where $\operatorname{Cov}(\cdot)$ denotes the covariance matrix.

The GSS-LSF estimator is the least-squares estimator
\begin{equation}
\hat{\boldsymbol{\alpha}} = (\mathbf A^\top\mathbf A)^{-1}\mathbf A^\top\mathbf m,
\end{equation}
which is unbiased and satisfies
\begin{equation}
\operatorname{Cov}(\hat{\boldsymbol{\alpha}}) = \sigma^2(\mathbf A^\top\mathbf A)^{-1}.
\end{equation}
Therefore, GSS-LSF attains the Cram\'er--Rao bound and is an efficient minimum-variance unbiased estimator.

Let $[\mathbf{F}^{-1}]_{pp}$ denote the $k\times k$ principal block of the inverse Fisher information matrix corresponding to source $p$, so that
\[
[\mathbf{F}^{-1}]_{pp} = \sigma^2[(\mathbf A^\top\mathbf A)^{-1}]_{pp}.
\]
Since $\hat{\mathbf x}_p = \mathbf B_p\hat{\boldsymbol{\alpha}}_p$,
\begin{equation}
\operatorname{Cov}(\hat{\mathbf x}_p) = \mathbf B_p[\mathbf{F}^{-1}]_{pp}\mathbf B_p^\top.
\end{equation}
Consequently,
\begin{equation}
\mathbb{E}\!\left[ \|\hat{\mathbf x}_p - \mathbf x_p^\star\|_2^2 \right]
= \operatorname{Tr}\!\left( \mathbf B_p[\mathbf{F}^{-1}]_{pp}\mathbf B_p^\top \right),
\end{equation}
for the GSS-LSF estimator, where $\mathbb{E}[\cdot]$ denotes expectation and $\operatorname{Tr}(\cdot)$ is the trace operator. The same quantity serves as a lower bound for any unbiased estimator. The bound depends on the latent input $\mathbf z$ through $\mathbf A$. When $\mathbf z$ is drawn from a continuous distribution (e.g., i.i.d.\ Gaussian), rank-deficient realizations occur with probability zero under mild conditions, thereby promoting identifiability of the spectral coefficients. If the true sources are not exactly contained in the truncated graph spectral subspaces, the model becomes misspecified and the estimator is generally biased.

\section*{Sensitivity to Hyperparameter}

We evaluate the sensitivity of GSS-LSF to the eigenvalue cutoff ratio $\lambda_{\text{ratio}}$, which controls the fraction of the low‑frequency spectrum retained for each source.
Figure~\ref{fig:ablation_smooth} reports the reconstruction SNR for both sources and their average.
The method remains robust over a wide range \(\lambda_{\text{ratio}} \in [0.1,\,0.7]\); the average SNR varies only between \(27.2\)\,dB and \(29.8\)\,dB, indicating that GSS-LSF is not overly sensitive to the precise choice of the cutoff.
As long as the essential low‑frequency support of the smooth signals is captured, the separation is reliable. When $\lambda_{\text{ratio}}$ is increased to \(0.8\), however, the average SNR drops sharply to \(19.7\)\,dB.
At such large ratios, the eigenbasis begins to include high‑frequency components that are dominated by noise.
Consequently, the unmixing process starts to fit the observation noise alongside each source, making it increasingly difficult to correctly identify the individual components.
This behavior confirms that for smooth signals, the source identification degrades as \(\lambda_{\text{ratio}}\) grows beyond the necessary low‑frequency band, because the separation algorithm inadvertently models noise rather than the underlying signal structure.
The stability up to \(\lambda_{\text{ratio}} = 0.7\) demonstrates that GSS-LSF can be deployed without meticulous hyperparameter tuning.

\begin{figure}
    \centering
    \includegraphics[width=0.7\linewidth]{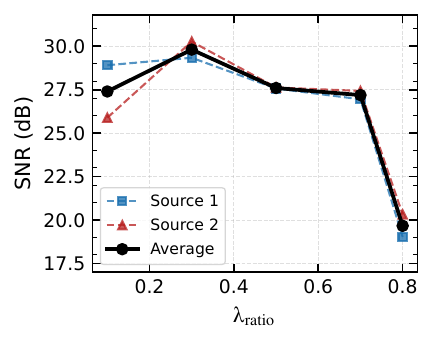}
    \caption{
        Ablation study of GSS-LSF on smooth and non-bandlimited signals.
    }
    \label{fig:ablation_smooth}
\end{figure}

\end{document}